\documentclass[12pt,english]{article}
\usepackage{lmodern}
\usepackage[T1]{fontenc}
\usepackage[latin9]{inputenc}
\usepackage{geometry}
\usepackage{color}
\usepackage{babel}
\usepackage{float}
\usepackage{url}
\usepackage{amsmath}
\usepackage{amsthm}
\usepackage{amssymb}
\usepackage{graphicx}
\usepackage{setspace}
\usepackage[authoryear]{natbib}
\usepackage[unicode=true,
 bookmarks=true,bookmarksnumbered=false,bookmarksopen=false,
 breaklinks=false,pdfborder={0 0 0},pdfborderstyle={},backref=false,colorlinks=true]
 {hyperref}
\hypersetup{
 pdfpagelayout=OneColumn, pdfnewwindow=true, pdfstartview=XYZ, plainpages=false, urlcolor=[rgb]{0.0430 ,0, 0.5}, linkcolor=[rgb]{0.0430 ,0, 0.5}, citecolor=[rgb]{0.0430 ,0, 0.5}, hypertexnames=false}

\makeatletter

\providecommand{\tabularnewline}{\\}

\theoremstyle{plain}
\newtheorem{assumption}{\protect\assumptionname}
\theoremstyle{remark}
\newtheorem{claim}{\protect\claimname}

\usepackage{chngcntr}
\setcitestyle{round}
\usepackage{mathtools}
\usepackage{breakcites}
\usepackage[all]{hypcap}
\usepackage{dcolumn}

\makeatother

\providecommand{\assumptionname}{Assumption}
\providecommand{\claimname}{Claim}

\begin{document}
\title{An Experiment on Network Density and \\ Sequential Learning\thanks{We thank the editors and two anonymous referees, J. Aislinn Bohren,
Jetlir Duraj, Ben Enke, Drew Fudenberg, Ben Golub, Jonathan Libgober,
Margaret Meyer, Matthew Rabin, Ran Spiegler, and Tomasz Strzalecki
for useful comments. Financial support from the Eric M. Mindich Research
Fund for the Foundations of Human Behavior is gratefully acknowledged.
Kevin He thanks the California Institute of Technology for hospitality
when some of the work on this paper was completed.}}
\author{Krishna Dasaratha\thanks{University of Pennsylvania. Email: \texttt{\protect\href{mailto:krishnadasaratha\%40gmail.com}{krishnadasaratha@gmail.com}}}
\and Kevin He\thanks{University of Pennsylvania. Email: \texttt{\protect\href{mailto:hesichao\%40gmail.com}{hesichao@gmail.com}}}}
\date{{\normalsize{}}%
\begin{tabular}{rl}
First version: & September 4, 2019\tabularnewline
This version: & May 19, 2021\tabularnewline
\end{tabular}}
\maketitle
\begin{abstract}
{\normalsize{}\thispagestyle{empty}
\setcounter{page}{0}}{\normalsize\par}

We conduct a sequential social-learning experiment where subjects
each guess a hidden state based on private signals and the guesses
of a subset of their predecessors. A network determines the observable
predecessors, and we compare subjects' accuracy on sparse and dense
networks. Accuracy gains from social learning are twice as large on
sparse networks compared to dense networks. Models of naive inference
where agents ignore correlation between observations predict this
comparative static in network density, while the finding is difficult
to reconcile with rational-learning models. 
\end{abstract}
\begin{flushleft}
{\small{}\newpage}{\small\par}
\par\end{flushleft}

\interfootnotelinepenalty=10000
\renewcommand*\&{and}

\section{\label{sec:Introduction}Introduction}

In many economic situations, people form beliefs based on others'
actions. In these settings, agents typically do not observe all members
of the society, but only a select subset \textemdash{} namely, their
neighbors in an underlying social network. How the structure of this
observation network affects learning outcomes is a fundamental question
for understanding social learning. While an extensive theoretical
literature has explored this question for both naive and rational
agents (e.g., \citealp*{golub2010naive,acemoglu2011bayesian,golub2012homophily}),
much less is known empirically.

Density is one of the most basic properties of a network. How do learning
patterns differ between sparse networks, where agents usually observe
very few neighbors, and dense networks, where agents generally have
abundant social information? On denser networks, agents observe more
predecessors (both directly and indirectly), so their actions can
incorporate the private signals of more individuals. But whether this
leads to more accurate learning ultimately depends on how society
aggregates these signals. Predecessors' actions can be correlated
by their common neighbors, so this aggregation may be difficult.

In this work, we conduct an experiment to compare social-learning
outcomes on sparse and dense networks. We study a sequential social-learning
environment where agents on an observation network each guess a hidden
state. We find that although later agents have fewer observations
on sparser networks, they nevertheless learn substantially better
on sparse networks than dense networks.

We place subjects into groups of 40 who act in order. Each group lives
on a social network, with randomly-generated links that determine
each subject's observations. Each subject has a 25\% chance of observing
each predecessor in the sparse treatment and a 75\% chance in the
dense treatment (and subjects know these probabilities). A hidden
binary state is drawn for each group. On her turn, each subject must
guess the state using her private signal and the past guesses of the
predecessors she observes. Subjects were paid for accuracy.

Prior to data collection, we pre-registered a measure of long-run
learning accuracy: the fraction of the final 8 subjects in the group
who correctly guess the state. Comparing this measure on 130 sparse
networks versus 130 dense networks, we find that denser networks lead
to worse learning accuracy. In dense networks, the average accuracy
of the last 8 subjects improves on the autarky benchmark (i.e., the
average accuracy if no one can observe others' actions) by 5.7\%,
but this improvement is 12.6\% in sparse networks. Thus, the long-run
accuracy gains from social learning are twice as large in the sparse
treatment as in the dense treatment ($p$-value 0.0239).

In addition to its direct implications about the role of network density
in social learning, this finding provides indirect evidence supporting
models of \emph{naive inference} in which agents neglect the correlations
among their social observations (as in \citealp{eyster2010naive}).
Motivated by a theoretical result from \citet{dasaratha2017network},
we compute predictions of the naive model. Later agents exhibit higher
accuracy on sparse networks than dense networks in this model, as
in our experimental evidence. The basic intuition is that an agent
with correlation neglect ends up placing too much weight on the actions
of the first few subjects in the same group, as these actions commonly
influence many of the agent's predecessors. When the network is denser,
this over-weighting is more severe and so naive agents' guesses are
less accurate in the long run.

On the other hand, our experimental findings are inconsistent with
the rational social-learning model. \citet*{acemoglu2011bayesian}'s
results imply that rational agents learn asymptotically in environments
matching our experimental setup. We adapt their methods to provide
lower bounds on the accuracy of rational agents 33 through 40 in the
sparse and dense treatments. These bounds imply that rational agents'
accuracy cannot improve substantially from the dense-network treatment
to the sparse-network treatment \textemdash{} in particular, the rational
model does not predict a doubling of accuracy gain.

Our data also show that network density has no statistically significant
effect on the \emph{overall} accuracy averaged across all 40 subjects
in each group. This is because dense networks increase the accuracy
of subjects who move early in the group, even though they lower the
accuracy of subjects who move later. This reversal of the accuracy
ranking between sparse and dense networks over the course of social
learning is another prediction of naive inference.

Finally, to provide additional evidence that learning is worse on
denser networks because subjects fail to account for correlation,
we conduct a variant of the experiment where subjects observe neighbors
who make conditionally independent guesses. The setup is the same
as in the main experiment, except the first $32$ agents in each group
only observe their own private signals, while the final 8 agents randomly
observe some of the initial $32$ agents. For the latter subjects,
average guess accuracy is $68.2\%$ when there is a $25\%$ chance
of observing each predecessor and $72.5\%$ when there is a $75\%$
chance of observing each predecessor. The extra observations in dense
networks improve guess accuracy when those observations are not correlated
by common social information.

\subsection{Related Literature}

Our experimental results add to a growing body of evidence that humans
do not properly account for correlations in social-learning settings.
\citet{enke2016correlation} show that correlation neglect is prevalent
even in simple environments where the observed information sources
are mechanically correlated. In a field experiment where agents interact
repeatedly with the same set of neighbors, \citet*{chandrasekhar2015testing}
find agents fail to account for redundancies.

Most closely related to the present work, the laboratory games in
\citet*{eyster2015experiment} and \citet*{mueller2015general} directly
evaluate behavioral assumptions matching ours. \citet*{eyster2015experiment}
find that on the complete observation network, many agents choose
the best response assuming predecessors are rational while some participants
exhibit redundancy neglect. On a more complex network the naive model
matches more observations than the rational model, and there is little
anti-imitation (which would be required for correct Bayesian inference,
as shown in \citealp{eyster2014extensive}).\footnote{In the complex network, four agents move in each period after observing
predecessors from previous periods.} \citet*{mueller2015general} find most observations are consistent
with the behavioral assumption we study (which they call quasi-Bayesian
updating) in a setting where agents have limited information about
the network. These experiments suggest naiveté may be more likely
in settings where agents either have a limited knowledge of the true
network or the network is known but very complicated. In these settings,
the correct Bayesian belief given one's observations can be far from
obvious, so agents are more likely to resort to behavioral heuristics.

Unlike this previous work, our experiment tests the comparative statics
predictions of naive and rational learning with respect to variations
in the learning environment. This allows us to cleanly test redundancy
neglect against rational updating. Our approach allows us to focus
on long-term learning outcomes\textemdash which are the welfare-relevant
metrics as we consider changes in the environment\textemdash instead
of solely on measuring individual behavior.

Several experiments in this literature, including \citet{grimm2014experiments},
\citet*{chandrasekhar2015testing}, and \citet*{mueller2015general},
test social learning outcomes under multiple network structures. In
these works, changes in network structure largely serve as a robustness
check for claims about subject behavior. By considering larger networks
and varying density, we show network structures play an important
role in learning outcomes and exploit this variation to better understand
behavior.

\section{\label{sec:experiment_theory} Theoretical Motivation}

\subsection{Model}

The state of the world $\omega\in\{0,1\}$ takes one of two possible
values with equal probabilities. The set of agents is indexed by $i\in\mathbf{\mathbb{N}}$.
Agents move in the order of their indices, each acting once.

On her turn, each agent $i$ observes a private signal $s_{i}\in\mathbb{R}$,
as well as the actions of some previous agents. Then, $i$ chooses
an action $a_{i}\in\{0,1\}$ to maximize the probability that $a_{i}=\omega$
given her belief about $\omega$.

Private signals $(s_{i})$ are i.i.d$.$ and Gaussian conditional
on the state of the world. When $\omega=1$, $s_{i}\sim\mathcal{N}(1,\sigma^{2})$.
When $\omega=0$, $s_{i}\sim\mathcal{N}(-1,\sigma^{2})$. Here $\sigma^{2}>0$
is the conditional variance of the private signal.

In addition to her signal, each agent $i$ observes the action of
each predecessor with probability $q$. These observations are i.i.d$.$
Independence of observations means that whether one agent observes
a certain predecessor does not depend on whether a different agent
observes the same predecessor. Agents observed by $i$ are called
the \emph{neighbors} of $i$, and the sets of neighbors define a (random)
directed network.

We compare two kinds of agents: rational agents and naive agents.
Rational agents play the unique perfect Bayesian equilibrium. Naive
agents optimize given the following misspecified beliefs:
\begin{assumption}
[Naive Inference Assumption]\textbf{ }\label{assu:behavioral}Each
agent wrongly believes that each predecessor chooses an action to
maximize her expected payoff based solely on her private signal, and
not on her observation of other agents.
\end{assumption}
Equivalently, naive agents believe that each of their neighbors observe
no other agents. Besides the error in Assumption \ref{assu:behavioral},
naive agents are otherwise correctly specified and optimize their
expected utility given their mistaken beliefs.

Assumption \ref{assu:behavioral} was introduced in a sequential-learning
setting where agents observe all predecessors by \citet{eyster2010naive}.
Their work refers to this form of inference as ``best-response trailing
naive inference'' (BRTNI).

\subsection{Naive and Rational Behavior}

\citet{dasaratha2017network} suggest an empirical test for the naive
inference assumption: in the context of sequential learning on uniform
random networks, does increasing the link-formation probability $q$
cause more inaccurate long-run beliefs? In this paper, we experimentally
test this comparative static in networks of $40$ agents by comparing
learning outcomes in sparse networks (where $q=\frac{1}{4}$) and
dense networks (where $q=\frac{3}{4}$).

The naive-learning model and the rational-learning model make competing
predictions about this comparative static. The intuition for naive
learning comes from \citet{dasaratha2017network}, which suggests
that overweighting due to correlation neglect is more severe on dense
networks.\footnote{\citet{dasaratha2017network} consider agents with a continuous action
space, but we implemented a binary action space in the experiment
for clarity. We felt it would be easier for subjects to make a binary
choice than to accurately report their exact belief.} We do not expect human subjects to behave exactly according to Assumption
\ref{assu:behavioral} \textemdash{} for example, the meta-analysis
of \citet{weizsacker2010we} reports that laboratory subjects in sequential
learning games suffer from autarky bias, underweighting their social
observations relative to the payoff-maximizing strategy. However,
the comparative static prediction of the naive model remains robust
even after introducing any fraction of autarkic agents.\footnote{See the Appendix of a previous version of \citet{dasaratha2017network},
available at \url{https://arxiv.org/pdf/1703.02105v5.pdf}.}

The prediction of the naive model is shown in Figure \ref{fig:Learning-on-Erdos-Renyi},
which plots the probabilities that each of the 40 naive agents will
correctly guess the state in sparse and dense networks with $\sigma=2.$
Because naive agents' actions only depend on the number of their predecessors
choosing each of the two actions and not the order of these actions,
recursively calculating the distributions of actions is computationally
feasible (see Appendix \ref{subsec:Performance-of-Naive} for details).
As shown in Figure \ref{fig:Learning-on-Erdos-Renyi}, early naive
agents do worse under $q=\frac{1}{4}$ than $q=\frac{3}{4}$ because
there is very little social information, but the comparison quickly
reverses as we examine later naive agents. 
\begin{figure}
\begin{centering}
\includegraphics[scale=0.6]{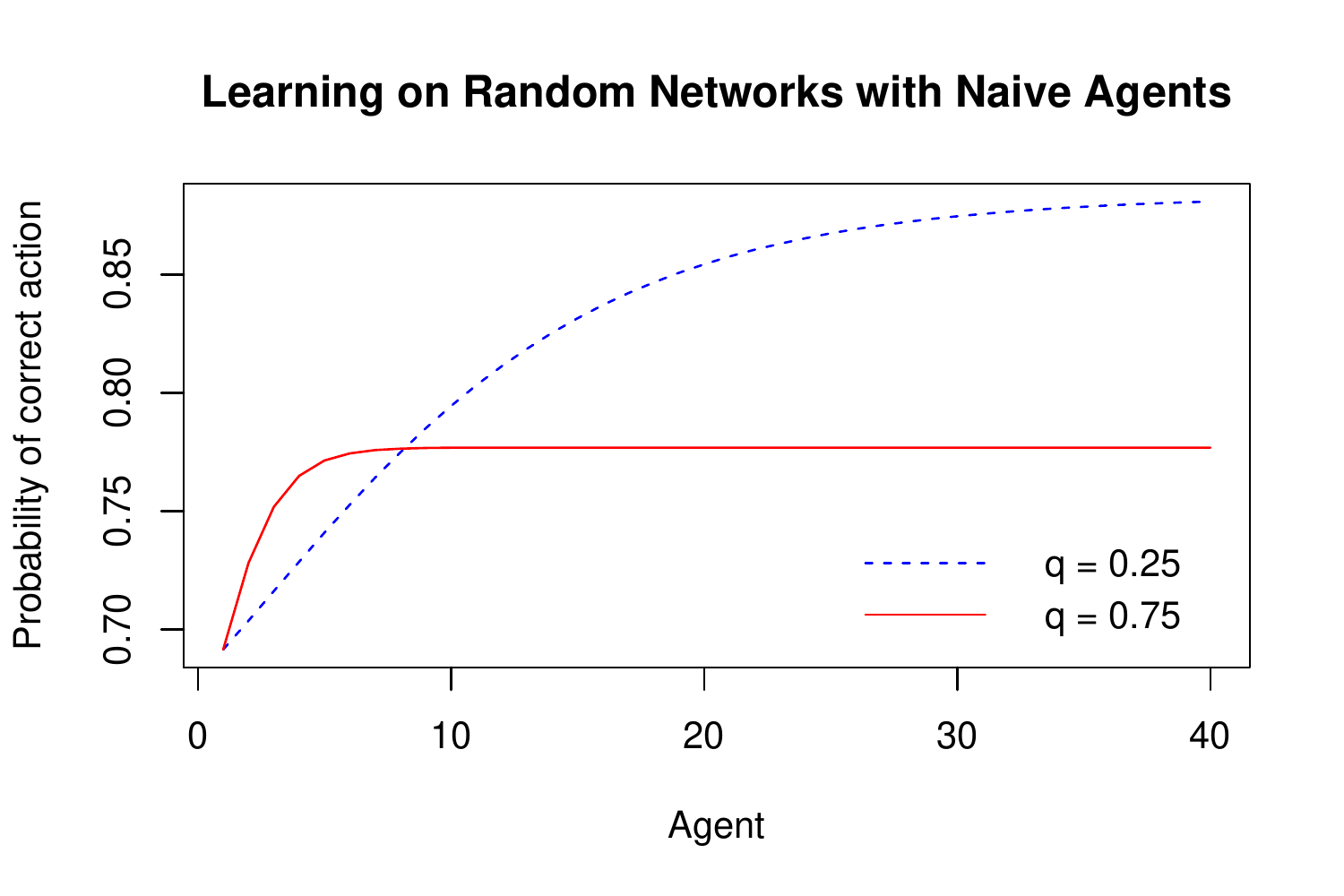}
\par\end{centering}
\caption{\label{fig:Learning-on-Erdos-Renyi}Learning accuracy on random networks
with $40$ naive agents, binary actions, and $\sigma^{2}=4$. Dashed
blue and solid red curves show the expected accuracy of different
agents on networks with link probabilities $q=\frac{1}{4}$ and $q=\frac{3}{4}$,
respectively.}
\end{figure}

On the other hand, the rational-learning model predicts that later
agents will have either similar or greater accuracy on the dense network
compared to the sparse network. \citet*{acemoglu2011bayesian}'s results
imply that in an environment matching our experimental setup, rational
agents will learn the true state in the long-run, regardless of the
network density. We can confirm that 40 rational agents are enough
to approach this asymptotic learning limit when $q=\frac{3}{4}$.
To do this, we compute a lower bound for the probability of correct
learning for each agent $i$ in the dense network of our experiment,
assuming all agents are rational Bayesians (see Appendix \ref{sec:Bounding-Performance-of}
for details). This lower bound is based on (suboptimal) agent strategies
that only depend on  own private signals and the action of just one
neighbor, as in the neighbor-choice functions in \citet*{lobel2015information}.
This exercise shows that the $33^{\text{rd}}$ rational agent is correct
at least 96.8\% of the time on dense networks, with the lower bound
on accuracy continuing to increase up to the $40^{\text{th}}$ agent,
who is correct at least $97.5\%$ of the time. In addition to suggesting
that the asymptotic result of \citet*{acemoglu2011bayesian} very
likely holds by the $40^{\text{th}}$ agent, the fact that this lower
bound for accuracy on the dense network is so close to perfect learning
proves the $40^{\text{th}}$ rational agent could not perform substantially
better on the sparse network,\footnote{We prove these bounds because we are not aware of a computationally
feasible method of calculating or simulating the probability that
rational agents are correct. \citet*{rahimian2014non} show computing
rational actions in another social learning environment is NP-hard.} contrary to the predicted improvement for the $40^{\text{th}}$ naive
agent shown in Figure \ref{fig:Learning-on-Erdos-Renyi}.

Intuitively one might also expect more connections to also help rational
agents in the short- and medium-run as they can adjust for potential
redundancies in information. For example, on the complete network
with continuous actions, rational agents can back out the private
signals of all predecessors by observing their actions, so every agent
$i$ does better on the complete network than on any sparser network
structure. We note, however, that exact comparative statics of the
rational model or variants are not known on random networks.

We experimentally test the competing predictions of the naive and
the rational models about how long-run accuracy varies with network
density. We thus provide indirect evidence for the naive inference
assumption, complementing the direct measurement of behavior in \citet*{eyster2015experiment}
and \citet*{mueller2015general}.

Beyond providing another form of evidence, our experiment also contributes
to understanding social learning by using the welfare-relevant outcome,
namely the long-run accuracy of actions, as the dependent variable.
Even if individual behavior tends to match redundancy neglect models
in simple or stylized settings, one might worry that the theoretical
implications of said models concerning aggregate learning need not
hold in practice for complex environments. For a policymaker who can
alter the observation network, for instance, experiments using welfare-relevant
outcomes as their dependent variables give more explicit guidance
as to the consequences of different policies.

\section{\label{subsec:Experimental-Design}Experimental Design}

We conducted our experiment on the online labor platform Amazon Mechanical
Turk (MTurk) using Qualtrics survey software.

We pre-registered our experimental protocol and regression specification
prior to the start of the experiment in August 2017. Our pre-registration
included the target sample size (which was met exactly) and the dependent
variable to measure the accuracy of social learning. The pre-registration
document can be found on the registry website at \texttt{\href{https://aspredicted.org/yp6eq.pdf}{https://aspredicted.org/yp6eq.pdf}}
and is also included in the Online Appendix.

We recruited $1040$ subjects. To be recruited, each subject must
correctly answer three comprehension questions (which were scenarios
in the game with a dominant choice). An additional $375$ MTurk users
incorrectly answered one or more comprehension questions and were
not allowed to participate in the experiment, based on the pre-registered
exclusion criteria. These excluded users were $26.5\%$ of the potential
subjects. The experiment was carried out in fall 2017.

In addition to comprehension questions, we restricted to subjects
located in the United States who had completed at least $50$ previous
MTurk tasks with a lifetime approval rate of at least $90\%$. Subjects
were not permitted to participate multiple times in the experiment.
There were at most $15$ subjects who did not complete all trials,
implying a completion rate of at least $98.5\%$. These non-completers
were excluded and replaced by new subjects.

Each \emph{trial} consisted of 40 agents who were asked to each make
a binary guess between two \emph{a priori }equally likely states of
the world, L (for left) and R (for right). The states were color-coded
to make instructions and observations more reader-friendly. Agents
are assigned positions in the sequence and move in order. Each MTurk
subject participated in 10 trials, all in the same position (depending
on when they participated in the experiment). The grouping of subjects
into trials was independent across trials. Subjects received $\$0.25$
for completing the experiment and $\$0.25$ per correct guess, for
a maximum possible payment of $\$2.75$. Subjects received no feedback
about the accuracy of their guesses until they were paid at the conclusion
of the experiment. Subjects ordinarily took less than $10$ minutes
to complete their participation and earned \$2.08 on average, so the
incentives were quite large for an MTurk task.

In each trial, every agent received a private signal, which had the
Gaussian distribution $\mathcal{N}(-1,4)$ in state L and the Gaussian
distribution $\mathcal{N}(1,4)$ in state R. These distributions were
presented visually in the instructions. Along with the value of their
signal, subjects were told the probability of each state conditional
on only their private signal.

Each trial was also associated with a density parameter, either $q=\frac{1}{4}$
or $q=\frac{3}{4}.$ A random network was generated for each trial
by linking each agent with each predecessor with probability $q$.
Each MTurk subject was assigned into either the ``sparse'' or the
``dense'' treatment, and then placed into 10 trials either all with
$q=\frac{1}{4}$ or all with $q=\frac{3}{4}.$ So there were 520 subjects
and 130 trials for each treatment. Agents were told the actions of
each linked predecessor and the link probability $q$ (but not the
full realized network, which could not be presented succinctly).

In each trial, agents viewed their private signal and any social observations
and were asked to guess the state. States, signals, and networks were
independently drawn across trials. Experimental instructions and an
example of a choice screen are shown in the Online Appendix.

\section{Results}

Let $y_{i,j}$ be the indicator random variable with $y_{i,j}=1$
if agent $i$ in trial $j$ correctly guesses the state, $y_{i,j}=0$
otherwise. Define $\tilde{y}_{j}:=\frac{1}{8}\sum_{i=33}^{40}y_{i,j}$
as the fraction of the last 8 agents in trial $j$ who correctly guess
the state. We test learning outcomes for the final 8 agents because
welfare depends on long-run learning outcomes in large societies and
these agents better approximate long-run outcomes. By using only her
private signal, an agent can correctly guess the state 69.15\% of
the time.\footnote{In fact, subjects in the first position (who have no social observations)
correctly use their private signals 93.8\% of the time.} We call $\tilde{y}_{j}-0.6915$ the \emph{gain from social learning}
in trial $j$, as this quantity represents improvement relative to
the autarky benchmark.

We find that the average gain from social learning is 8.73 percentage
points for the $q=\frac{1}{4}$ treatment and 4.12 percentage points
for the $q=\frac{3}{4}$ treatment. Social learning improves accuracy
on the sparse networks by twice as much as on the dense networks.
To test for statistical significance, we consider the regression 
\[
\tilde{y}_{j}=\beta_{0}+\beta_{1}q_{j}+\epsilon_{j}
\]
where $q_{j}\in\{\frac{1}{4},\frac{3}{4}\}$ is the network density
parameter for trial $j$. Recall that each subject was assigned into
ten random trials with the same network density and in the same sequential
position. This means for two different trials $j^{'}\ne j^{''}$,
the error terms $\epsilon_{j^{'}}$ and $\epsilon_{j^{''}}$ are close
to independent since there are likely very few subjects who participated
in both trials.

We estimate $\beta_{1}=-0.092$ with a $p$-value of 0.0239 (see Table
\ref{tab:Regression}). The results are the same whether we use robust
standard errors or not. These findings are consistent with naive updating
but not with rational updating, as discussed in Section \ref{sec:experiment_theory}.\footnote{We pre-registered average accuracy in the last 8 agents (i.e last
20\% of agents) as the dependent variable for the experiment, but
the regression result is robust to other definitions of $\tilde{y}_{j}$.
When $\tilde{y}_{j}$ encodes average accuracy among the last $m$
agents for any $4\le m\le12$ (i.e. between last 10\% and last 30\%
of the agents), the estimate for $\beta_{1}$ remains negative.}

\begin{table}
\begin{center} \begin{tabular}{l*{2}{c}} \hline\hline                                          &\multicolumn{1}{c}{FractionCorrect}\\ \hline NetworkDensity      &         -0.0923\\                     &   (0.0406)\\ [1em] Constant            &       0.802\\                     &       (0.0218)\\ \hline Observations        &             260\\ Adjusted \(R^{2}\)  &        0.016\\ \hline\hline \end{tabular} \end{center}

\caption{\label{tab:Regression} Regression results for the effect of network
density on learning outcomes (with robust standard errors).}
\end{table}

This difference in the gains from social learning is not driven by
different rates of autarky among the two treatments for the last 8
agents. We say an agent \emph{goes against her signal} if she guesses
L when her signal is positive or guesses R when her signal is negative.
Within the last 8 rounds, there are 138 instances of agents going
against their signals in the $q=\frac{1}{4}$ treatment, which is
very close to the 136 instances of the same under the $q=\frac{3}{4}$
treatment. However, when agents go against their signals in the last
8 rounds, they correctly guess the state 81.88\% of the time under
the $q=\frac{1}{4}$ treatment, but only 71.32\% of the time under
the $q=\frac{3}{4}$ treatment. This shows the observed difference
in accuracy is due to social learning being differentially effective
on the two network structures.

However, the $q=\frac{3}{4}$ treatment yields better learning outcomes
for early agents. For agents 10 through 20, the average guess accuracy
is 72.24\% under the $q=\frac{1}{4}$ treatment and 73.22\% under
the $q=\frac{3}{4}$ treatment. As such, if we replace the dependent
variable in the pre-registered regression with overall accuracy $\bar{y}_{j}:=\frac{1}{40}\sum_{i=1}^{40}y_{i,j}$,
then we do not find a statistically significant estimate for $\beta_{1}$
($p$-value of 0.663). This result is consistent with the naive-learning
model: according to the predictions of the naive model shown in Figure
\ref{fig:Learning-on-Erdos-Renyi}, early agents are more accurate
under $q=\frac{3}{4}$, but later agents are more accurate under $q=\frac{1}{4}.$
The point of overtaking happens at a later round in practice than
in theory, because our experimental subjects rely more on their private
signal than predicted by the naive model,\footnote{The overall frequency of agents going against their signals was 36.8\%
of the predicted frequency under the naive model.} consistent with the meta-analysis of \citet{weizsacker2010we}.

Our experiment was designed to compare long-run learning accuracy
on different networks instead of measuring individual behavior. We
do not directly test alternate behavioral models for two reasons.
First, given the complex signal and network structures, such tests
will be very noisy in our data. Second, because the spaces of possible
networks and actions have very high dimension, it is computationally
infeasible to determine the action that each agent would take under
common knowledge of rationality. However, in the next subsection we
provide some evidence that our findings are driven by herding under
naive inference rather than other behavioral mechanisms.

\subsection{Evidence of naive herding}

In this section, we present three pieces of evidence suggesting that
naive herding is the mechanism responsible for the difference in learning
accuracy between the two treatments.

\textbf{(1)} \textbf{Distribution of overall accuracy}. Figure \ref{fig:Histograms}
in Appendix \ref{sec:figs_and_tables} plots the distributions of
subjects who correctly guess the state in the $q=\frac{1}{4}$ and
$q=\frac{3}{4}$ treatments, across different trials. Compared to
the distribution under $q=\frac{1}{4}$, the distribution under $q=\frac{3}{4}$
has more extreme values and a larger standard deviation (11.36 percentage
points versus 9.12 percentage points). This is suggestive evidence
for naive herding. With denser networks, we simultaneously find more
trials where agents do very badly overall (from herding on the wrong
state) and more trials where agents do very well overall (from herding
on the correct state).

\textbf{(2) Effect of misleading early signals on the accuracy of
later agents}. Call a private signal \emph{misleading} if it is positive
while the state is L, or if it is negative while the state is R. If
naive herding is the mechanism, we would expect misleading signals
received by early agents to be more harmful for eventual learning
accuracy on denser networks than on sparser networks. On the other
hand, a different behavioral mechanism based on the salience of the
visible decisions would suggest that early misleading signals are
more harmful on sparse networks, since each visible decision is more
salient when agents have fewer social observations. To test the naive
herding mechanism, we expand our baseline regression to include two
additional regressors: the number $m_{j}$ of the first fifth of agents
who receive misleading signals in trial $j$, and its interaction
effect with network density. That is, we estimate
\[
\tilde{y}_{j}=\beta_{0}+\beta_{1}q_{j}+\beta_{2}m_{j}+\gamma(q_{j}m_{j})+\epsilon_{j}.
\]
The difference in the marginal effect of a misleading early signal
for learning accuracy on the dense network $(q=\frac{3}{4})$ versus
on the sparse network $(q=\frac{1}{4})$ is $\frac{1}{2}\gamma$ in
the above specification.

As reported in Table \ref{tab:Regression_misleading} in Appendix
\ref{sec:figs_and_tables}, we find $\gamma=0.05$ with a $p$-value
of 0.0923. This means each misleading signal among the first fifth
of agents harms the average accuracy of the last fifth of agents in
the same trial by an extra 2.5 percentage points in dense networks
compared to sparse networks.

\textbf{(3) Average uncertainty}. Based on simulation evidence, we
expect naive agents to exhibit more agreement on denser networks.
To test this prediction in the data, we consider for each trial a
set of 30 moving windows centered around periods 6, 7, ... 35, with
each window spanning 11 consecutive periods. For each trial $j$ and
each window $w$, we compute $r_{j,w}\in\{0,\frac{1}{11},...,1\}$
as the fraction of 11 agents in the window who guessed R, and we let
$u_{j,w}:=r_{j,w}\cdot(1-r_{j,w})$ be a measure of uncertainty within
the window.\footnote{The value of $u_{j,w}$ would be unchanged if we instead defined $r_{j,w}$
as the fraction of the 11 agents in window $w$ who correctly guessed
the state.} In windows where agents exhibit a greater degree of agreement, we
will see a lower $u_{j,w}$. Under herding, we expect lower uncertainty
on denser networks, as higher density accelerates convergence to a
(possibly mistaken) social consensus. We find in the data that the
average  uncertainty across all trials and all windows is 0.165 on
dense networks and 0.178 on sparse networks. Examining uncertainty
in each of the 30 windows $w$ separately, we find average $u_{j,w}$
across trials is lower among dense networks than sparse networks for
all but 1 out of 30 windows. Numerically, the naive herding theory
predicts lower average $u_{j,w}$ on denser networks in all 30 windows.

\section{Neighbors with Conditionally Independent Actions}

In our main experiment, we find that denser networks lead to worse
social learning by later subjects. We have presented evidence suggesting
the mechanism behind this result is that subjects neglect correlation
in observed actions. To provide additional evidence for this channel,
we now test how network density affects social learning when observed
actions are conditionally independent given the state. In this section,
we will ask whether more observations help subjects whose neighbors
only have private information.

\subsection{Experimental Design}

We also pre-registered the experimental protocol and regression specification
for this second experiment, including the dependent variable to measure
the accuracy of social learning and the target sample size, prior
to the start of the experiment in November 2020. The pre-registration
document is included in the Online Appendix and may also be accessed
via the registry website at \texttt{\href{https://aspredicted.org/ag8fr.pdf}{https://aspredicted.org/ag8fr.pdf}}.

This experiment was also conducted online on MTurk. We recruited $624$
subjects, and each subject participated in $10$ trials. There were
a total of $130$ trials. To increase power, each trial included subjects
in both sparse and dense treatments. The first $32$ subjects in each
trial had no neighbors, and chose actions based only on their private
signals. Each trial also contained $8$ subjects in the sparse treatment
and $8$ subjects in the dense treatment. Subjects in the sparse treatment
observed each of the first $32$ subjects in the same trial with probability
$q=\frac{1}{4}$ while subjects in the dense treatment observed each
of the first $32$ subjects with probability $q=\frac{3}{4}$. There
were no other observations, so the actions of the observed neighbors
are always uncorrelated given the state. In particular, the subjects
after the first 32 in each trial never observe each other.

We maintained the state distribution, private signal distribution,
and action space from the main experiment. Recruitment and payment
were also the same as in the main experiment. The experimental instructions
were modified to accurately describe the social information subjects
would receive, if any. The first $32$ subjects in each trial (like
the first subject in each trial in the main experiment) were only
asked the one comprehension question that just involves private signals,
as the other comprehension questions pertain to subjects who receive
social information. Subjects earned an average of \$1.90 per person
in this second experiment.

\subsection{Results}

We find that the average accuracy is $68.2\%$ in the sparse treatment
and $72.5\%$ in the dense treatment. When subjects' neighbors only
have private information and not social information, having more neighbors
improves the accuracy of guesses.

In each trial, we will index the $8$ subjects in the sparse treatment
as $33,\hdots,40$ and the $8$ subjects in the dense treatment as
$41,\hdots,48$. Let $y_{i,j}$ be the indicator random variable with
$y_{i,j}=1$ if agent $i$ in trial $j$ correctly guesses the state,
$y_{i,j}=0$ otherwise. For each $q\in\left\{ \frac{1}{4},\frac{3}{4}\right\} $,
we define $\tilde{y}_{j}^{q}$ as the fraction of the 8 subjects in
that treatment in trial $j$ who correctly guess the state, so 
\[
\tilde{y}_{j}^{\frac{1}{4}}:=\frac{1}{8}\sum_{i=33}^{40}y_{i,j}\text{ and \ensuremath{\tilde{y}_{j}^{\frac{3}{4}}}:=\ensuremath{\frac{1}{8}\sum_{i=41}^{48}y_{i,j}}}.
\]

To test for statistical significance, we consider the regression 
\[
\tilde{y}_{j}^{q}=\beta_{0}^{uncor}+\beta_{1}^{uncor}q+\epsilon_{j,q}
\]
where $q\in\{\frac{1}{4},\frac{3}{4}\}$ is the network density parameter.
We estimate $\beta_{1}^{uncor}=0.087$ with a $p$-value of $0.0391$
(see Table \ref{tab:Regression-Independent}).

\begin{table}
\begin{center} \begin{tabular}{l*{2}{c}} \hline\hline                                          &\multicolumn{1}{c}{FractionCorrect}\\ \hline NetworkDensity      &     0.0865    \\                     &   (0.0417)\\ [1em] Constant            &       0.660\\                     &       ( 0.0229)\\ \hline Observations        &             260\\ Adjusted \(R^{2}\)  &        0.013\\ \hline\hline \end{tabular} \end{center}

\caption{\label{tab:Regression-Independent} Regression results for the effect
of network density on learning outcomes for subjects observing neighbors
with only private information (with robust standard errors).}
\end{table}

The difference in average accuracy is again driven by a difference
in the value of social information. Recall that a subject goes against
her signal if her signal is positive and she chooses L or her signal
is negative and she chooses R. Conditional on going against one's
own signal, subjects correctly guess the state $53.66\%$ of the time
in sparse treatment and $69.39\%$ of the time in dense treatment.

Guesses are in general less accurate in this follow-up experiment
than in the main experiment. The subjects in the first $32$ positions
in each trial had only one comprehension question because their decision
problems did not involve any social information. Subjects who did
not fully understand the experimental instructions may therefore have
been more likely to participate in the experiment in these positions,
producing much noisier choices that degrade later subjects' accuracy.\footnote{Subjects in the first 32 positions correctly used their private signals
only 81.7\% of the time.} There may also be differences in the MTurk subject pool compared
to the main experiment, as the second experiment was conducted three
years later.

The follow-up experiment finds that having more observations improves
accuracy when those observations are conditionally uncorrelated. This
provides additional evidence that our main result is driven by the
failure of subjects to account for correlation in observed actions,
rather than by some other mechanism that does not depend on this correlation.

\section{Concluding Discussion}

Our study provides experimental evidence on how the density of the
observation network affects people's long-run accuracy in social-learning
settings. We find that sparser networks double the accuracy gains
from social learning relative to denser networks. While the rational
model predicts correct asymptotic social learning with minimal assumptions
on the social network, we conjecture that in practice, many structural
properties of the network can substantially alter long-run accuracy.
Our empirical findings support this conjecture for the case of network
density, one of the most canonical network statistics. We leave open
the roles of other network structures as promising future work.

We have argued that our experimental results provide evidence for
inferential naiveté by analyzing a particular form of behavior (Assumption
\ref{assu:behavioral}). We conclude by discussing two ways in which
the experimental results are potentially consistent with more general
models of behavior. First, we have discussed models where all agents
are rational or all agents are naive, but a model where only some
of the agents suffer from inferential naiveté may be more realistic.
Such a model could also generate herding on incorrect beliefs, and
this herding may be more likely on denser networks. The exact details
depend on how the agents who do not suffer from inferential naiveté
reason about others' play. If these agents wrongly believe that others
are playing the perfect Bayesian equilibrium strategies, then they
will fail to correct the mistakes of naive agents. In this case, early
agents' actions can have very disproportionate influence on later
agents.

Second, Assumption \ref{assu:behavioral} is a particular form of
naive updating that assumes agents entirely neglect correlations in
neighbors' actions. Even in homogeneous populations, intermediate
forms of naive updating could also generate herding on incorrect beliefs.
Our main result suggests inferential naiveté, but does not distinguish
between alternate naive models involving some correlation neglect.

\bibliographystyle{ecta}
\bibliography{network_naive_sequential}

\appendix
\begin{center}
\textbf{\Large{}Appendix}{\Large\par}
\par\end{center}

\section{Theoretical Predictions in the Experimental Environment}

\subsection{Bounding the Performance of Rational Agents\label{sec:Bounding-Performance-of}}

Consider 40 rational agents on a random network where each agent is
linked to each of her predecessors $\frac{3}{4}$ of the time, i.i.d.
across link realizations. Agents know their own neighbors but have
no further knowledge about the realization of the random network.
The signal structure and payoff structure match the experimental design
in Section \ref{subsec:Experimental-Design}.

We provide a lower bound for the accuracy of agents 33 through 40
in the unique PBE of the social-learning game. We first show that
when every player uses the equilibrium strategy, all agents learn
at least as well as when everyone uses any \emph{constrained strategy}
that chooses an action based on only own private signal and the action
of the most recent neighbor. We then exhibit payoffs under one such
strategy, which give a lower bound on rational performance.

Fix an arbitrary sequence of constrained strategies $(\sigma_{i})$
where $\sigma_{i}:S_{i}\times\{0,1,\emptyset\}\to\Delta(\{0,1\})$
is only a function of $i$'s signal $s_{i}$ and the action of the
most recent predecessor that $i$ observes ($\sigma_{i}(s_{i},\emptyset)$
refers to $i$'s play if $i$ does not observe any predecessor). Let
$a_{i}$ denote $i$'s (random) action induced by this sequence of
strategies. Let $a'_{i}$ denote $i$'s (random) action when all agents
use the PBE strategy.
\begin{claim}
For all $i$, $\mathbb{P}[a'_{i}=\omega]\mathbb{\geq P}[a_{i}=\omega].$
\end{claim}
\begin{proof}
The proof is by induction on $i$ and the base case of $i=1$ is clear.
Suppose the claim holds for $i=1,...,n$. Conditional on agent $n+1$
observing no predecessors, the claim again holds as in the base case,
so we can check the claim conditional on $n+1$ observing at least
one neighbor.

Let $j$ be the most recent neighbor that $n+1$ observes. Then the
rational agent observes $s_{n+1}$, $a'_{j}$ for some $j\le n$,
and perhaps some other actions while the constrained agent only uses
$s_{n+1}$ and $a_{j}$ in decision-making, where $\mathbb{P}[a'_{j}=\omega]\mathbb{\geq P}[a_{j}=\omega]$
by the inductive hypothesis. By garbling the observed action $a'_{j}$,
the rational agent could construct a random variable with the same
joint distribution with $\omega$ as the less accurate action $a_{j}$.
Ignoring information other than $s_{n+1}$ and the garbled $a'_{j},$
the rational agent $n+1$ could therefore follow a strategy that does
as well as agent $n+1$ under the strategy profile $(\sigma_{i})$.
So we must have $\mathbb{P}[a'_{n+1}=\omega]\mathbb{\geq P}[a_{n+1}=\omega]$
when everyone uses the PBE strategy.
\end{proof}
We then numerically compute the values for $\mathbb{P}[a_{_{i}}=\omega]$
under the optimal constrained strategy, which are displayed in Table
\ref{fig:Lower-bounds}.

\begin{table}[H]
\begin{centering}
\begin{tabular}{|c|}
\hline 
agent number\tabularnewline
\hline 
\hline 
probability correct\tabularnewline
\hline 
\end{tabular}%
\begin{tabular}{|c|c|c|c|c|c|c|c|}
\hline 
33 & 34 & 35 & 36 & 37 & 38 & 39 & 40\tabularnewline
\hline 
\hline 
0.9685 & 0.9695 & 0.9705 & 0.9714 & 0.9723 & 0.9731 & 0.9739 & 0.9746\tabularnewline
\hline 
\end{tabular}
\par\end{centering}
\caption{\label{fig:Lower-bounds}Lower bounds on the accuracy of rational
agents on dense networks.}
\end{table}

\subsection{Performance of Naive Agents\label{subsec:Performance-of-Naive}}

Consider $40$ naive agents on a random network where each agent is
linked to each of her predecessors with probability $q$, i.i.d. across
link realizations. The signal structure and payoff structure match
the experimental design in Section \ref{subsec:Experimental-Design}.

We will compute the accuracy of each agent by a recursive calculation.
Because naive agents' actions do not depend on the order of predecessors,
behavior depends only on the number of agents who have played L and
the number of agents who have played R as well as the network. We
will compute the distribution over the number of agents from the first
$n$ who have played L and the number who have played R recursively.

Assume the state is R. Let $P(k,k')$ be the probability that $k$
of the first $n$ agents play L and $k'$ of the first $n$ agents
play R. We define $P(k,k')=0$ if $k<0$ or $k'<0.$ The posterior
log-likelihood of state R for a naive agent observing one action equal
to R (and no signal) is
\[
\ell=\frac{2}{\sigma^{2}}\cdot\frac{\mu+\sigma\phi(-\mu/\sigma)}{1-\Phi(-\mu/\sigma)},
\]
where $\Phi$ and $\phi$ are the distribution function and probability
density function of a standard Gaussian random variable, respectively.

Then we have the recursive relation
\begin{eqnarray*}
P(k,k') & =P(k-1,k')\sum_{i\leq k-1,i'\leq k'}B(i,k-1,q)B(i',k',q)\Phi(\frac{\sigma(i-i')\ell-2\mu\sigma}{2})+\\
 & P(k,k'-1)\sum_{i\leq k,i'\leq k'-1}B(i,k,q)B(i',k'-1,q)[1-\Phi(\frac{\sigma(i-i')\ell-2\mu\sigma}{2})],
\end{eqnarray*}
where $B(i,k,q)$ is the probability a binomial distribution with
parameters $k$ and $q$ is equal to $i$. The first summand gives
the probability of agent $k+k'$ choosing L after $k-1$ predecessors
choose L and the remainder choose R, and the second summand gives
the probability of agent $k+k'$ choosing R after $k$ predecessors
choose L and the remainder choose R. The binomial coefficients correspond
to the possible network realizations. Here we use naive inference,
which implies that only the number of observed agents choosing each
action matters for behavior and not their order.

From these distributions $P(\cdot,\cdot)$ we can compute the probability
that agent $n$ chooses the correct action R:
\[
\sum_{k=0}^{n}P(k,n-k)\sum_{i\leq k,i'\leq n-k}B(i,k,q)B(i',n-k,q)[1-\Phi(\frac{\sigma(i-i')\ell-2\mu\sigma}{2})].
\]
These probabilities, which we compute numerically, are displayed in
Table \ref{fig:Naive-values} for agents $33$ through $40$. 
\begin{table}[H]
\begin{centering}
\begin{tabular}{|c|}
\hline 
agent number\tabularnewline
\hline 
\hline 
accuracy with $q=1/4$\tabularnewline
\hline 
accuracy with $q=3/4$\tabularnewline
\hline 
\end{tabular}%
\begin{tabular}{|c|c|c|c|c|c|c|c|}
\hline 
33 & 34 & 35 & 36 & 37 & 38 & 39 & 40\tabularnewline
\hline 
\hline 
0.8773 & 0.8780 & 0.8786 & 0.8792 & 0.8797 & 0.8801 & 0.8805 & 0.8808\tabularnewline
\hline 
0.7768 & 0.7768 & 0.7768 & 0.7768 & 0.7768 & 0.7768 & 0.7768 & 0.7768\tabularnewline
\hline 
\end{tabular}
\par\end{centering}
\caption{\label{fig:Naive-values}The accuracy of naive agents on sparse and
dense networks.}
\end{table}

\section{\label{sec:figs_and_tables} Relegated Figures and Tables}

\begin{figure}[H]
\begin{centering}
\includegraphics[scale=0.8]{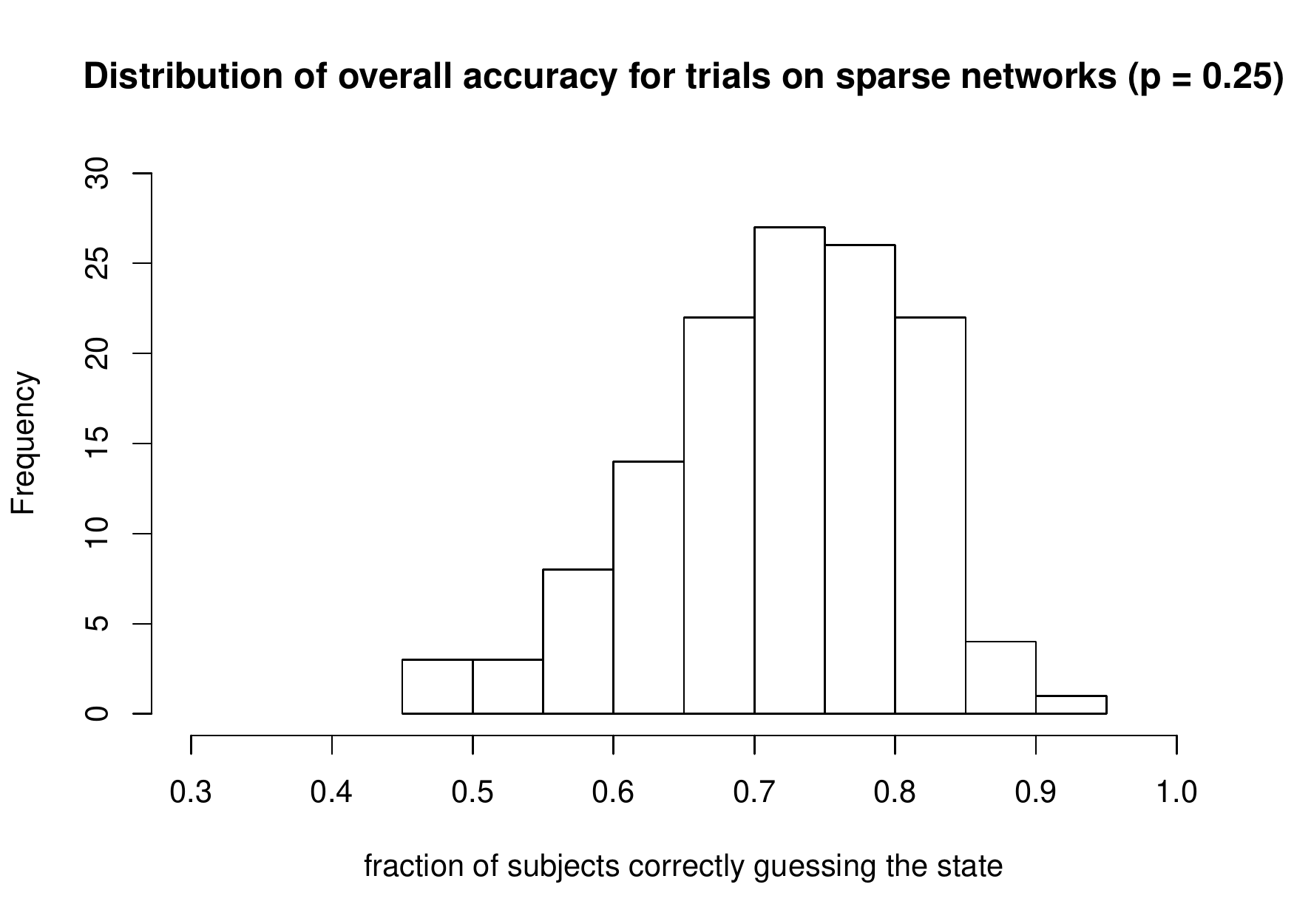} \includegraphics[scale=0.8]{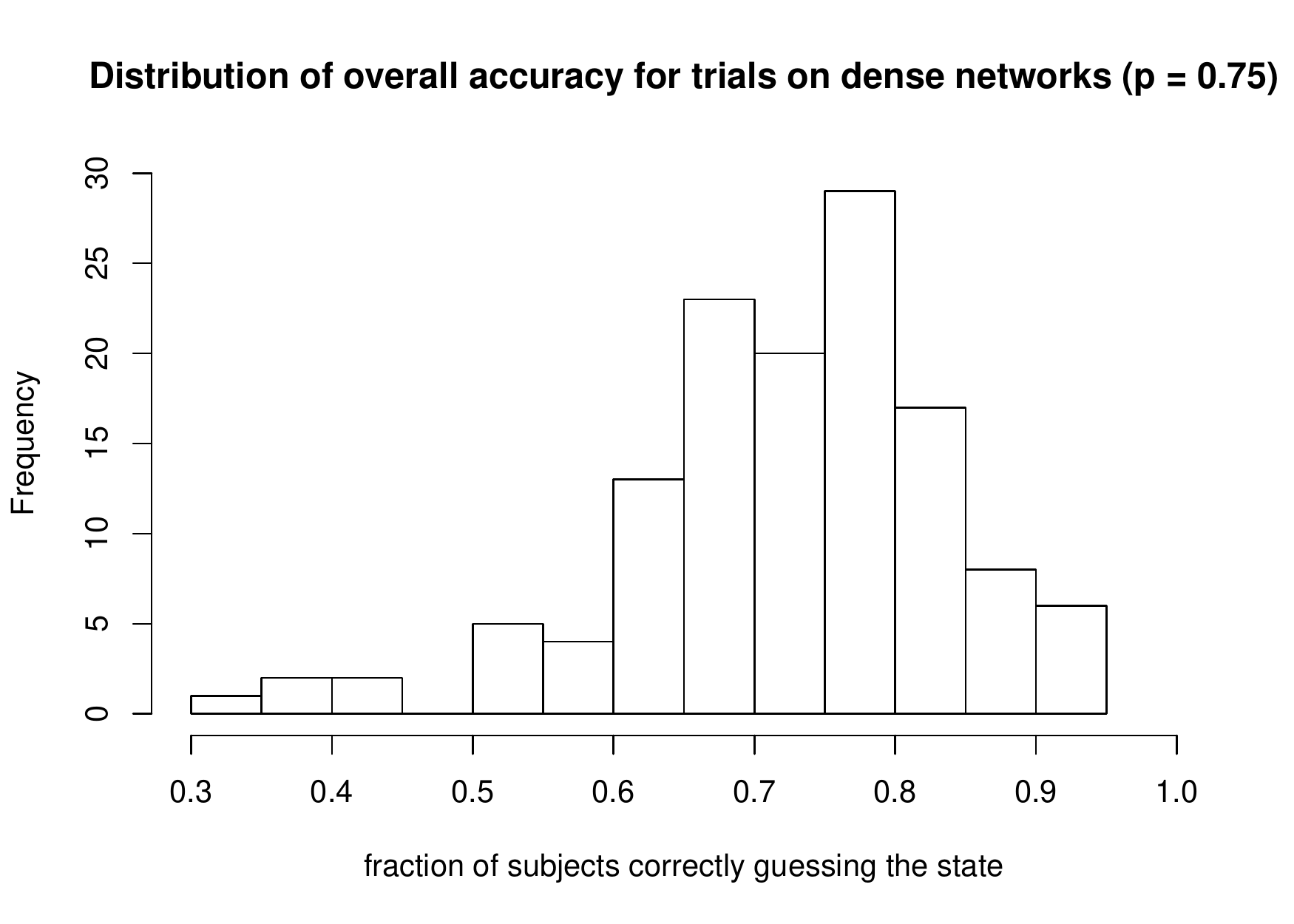}
\par\end{centering}
\caption{\label{fig:Histograms}Histograms of fractions of agents correctly
guessing the state}
\end{figure}

\begin{table}[H]
\begin{center} \begin{tabular}  {@{\extracolsep{5pt}}lc}  \\[-1.8ex]\hline  \hline \\[-1.8ex]   & \multicolumn{1}{c}{\textit{Dependent variable:}} \\  \cline{2-2}  \\[-1.8ex] & FractionCorrect \\  \hline \\[-1.8ex]   MisleadingEarlySignals & 0.014 \\    & (0.017) \\    & \\   NetworkDensity & 0.033 \\    & (0.082) \\    & \\   MisleadingEarlySignals$\times$NetworkDensity & $-$0.050$^{*}$ \\    & (0.030) \\    & \\   Constant & 0.768$^{***}$ \\    & (0.045) \\    & \\  \hline \\[-1.8ex]  Observations & 260 \\  R$^{2}$ & 0.040 \\  Adjusted R$^{2}$ & 0.029 \\  Residual Std. Error & 0.163 (df = 256) \\  F Statistic & 3.566$^{**}$ (df = 3; 256) \\  \hline  \hline \\[-1.8ex]  \textit{Note:}  & \multicolumn{1}{r}{$^{*}$p$<$0.1; $^{**}$p$<$0.05; $^{***}$p$<$0.01} \\   \end{tabular}  \end{center}

\caption{\label{tab:Regression_misleading} Effect of misleading early signals.}
\end{table}

\newpage{}
\begin{center}
\textbf{\Large{}Online Appendix}{\Large\par}
\par\end{center}

\section{Experimental Instructions\label{subsec:Appendix-Experimental-Design}}

Instructions and an example choice follow. To avoid confusion, the
instructions were modified for player $1$ in each round to exclude
discussion of social observations. A sample experiment can be completed
online at \texttt{\href{https://upenn.co1.qualtrics.com/jfe/form/SV_42dq2J2wHO30zA1}{https://upenn.co1.qualtrics.com/jfe/form/SV\_42dq2J2wHO30zA1}}

\includegraphics[scale=0.65]{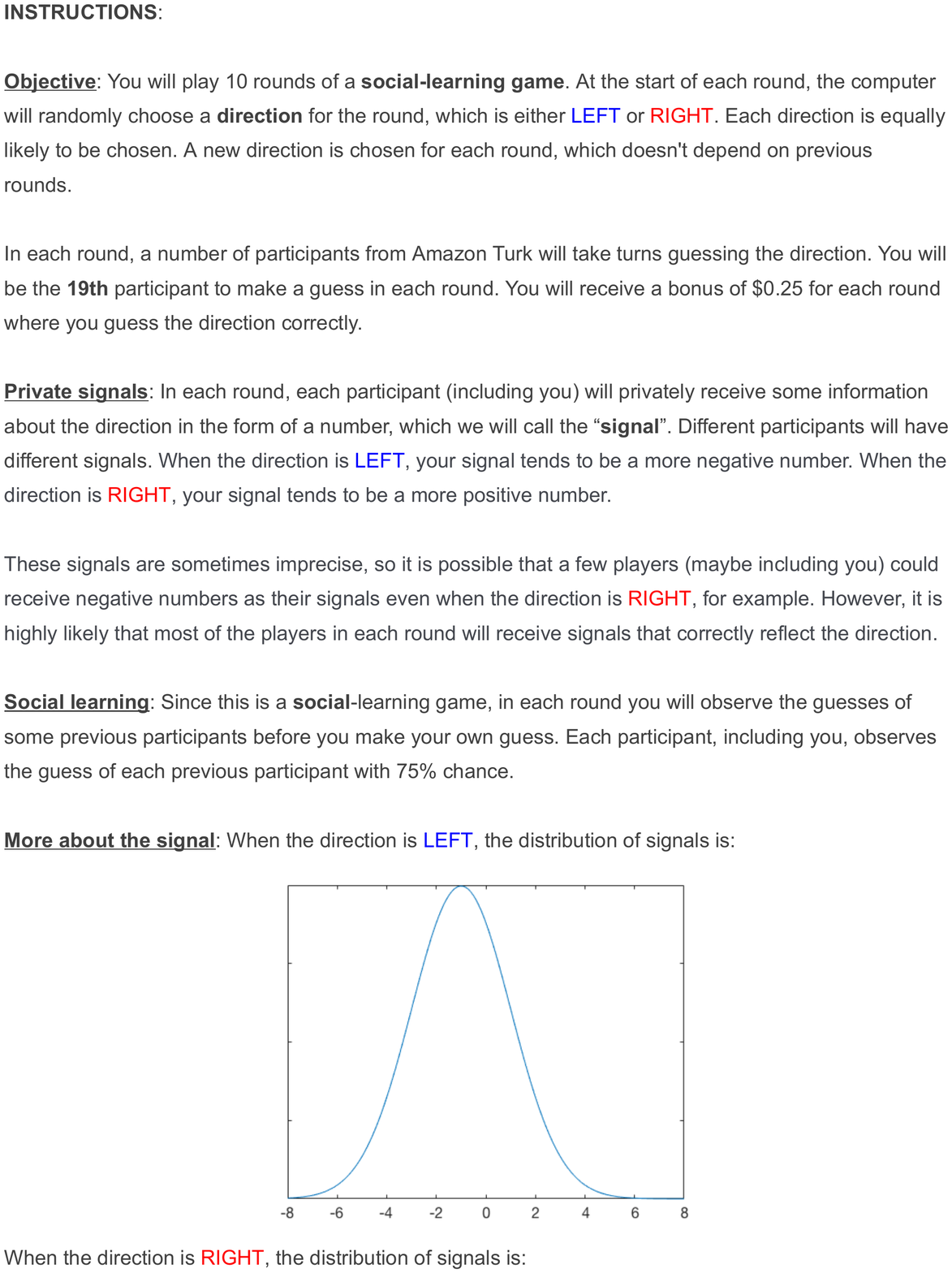}\newpage\includegraphics[scale=0.75]{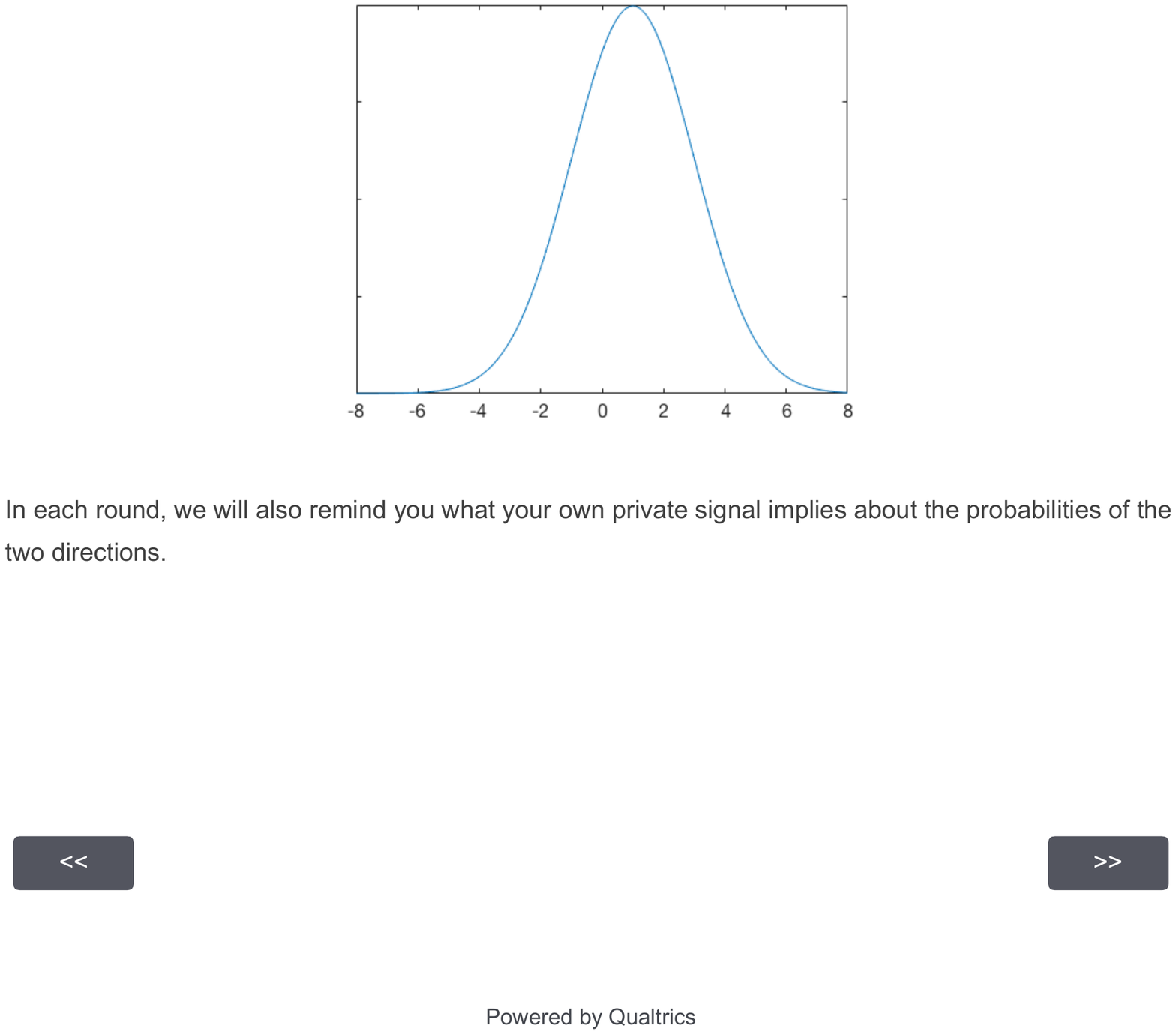}

\includegraphics[scale=0.7]{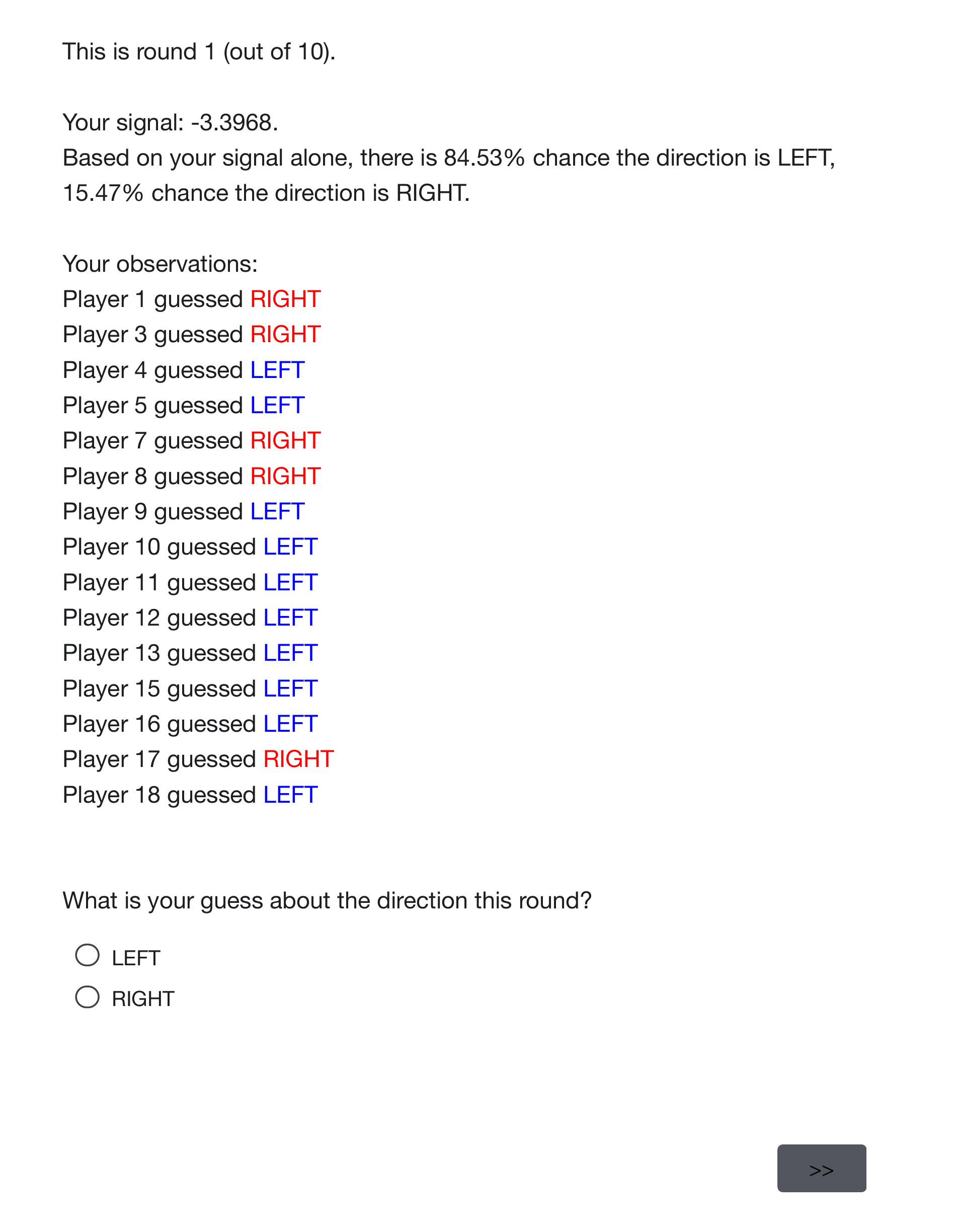}

\newpage{}

\section{Pre-Registration Documents \label{subsec:Pre-Registration-Document}}
\begin{center}
\includegraphics[scale=0.6]{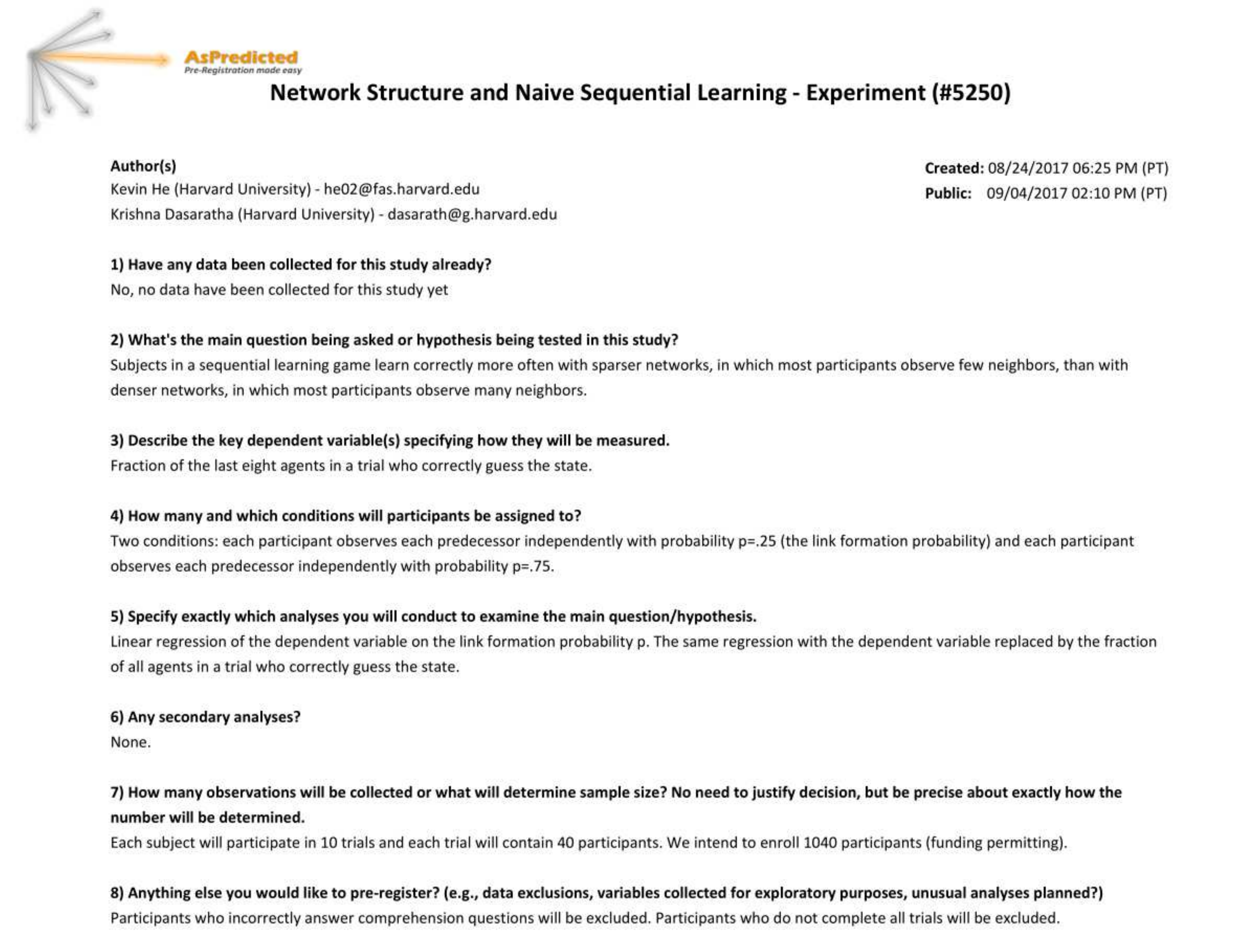}
\par\end{center}

\begin{center}
\includegraphics[scale=0.85]{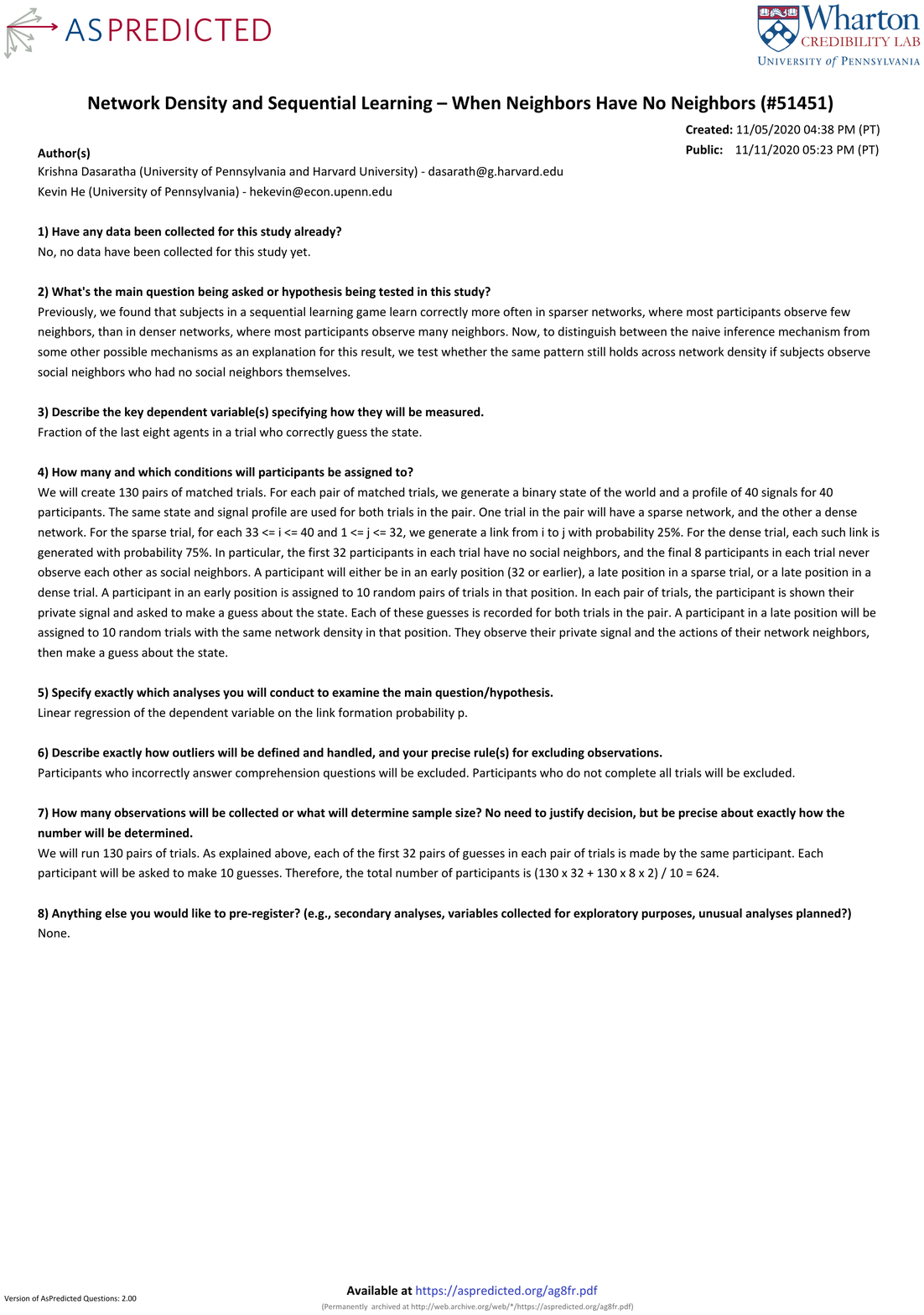}
\par\end{center}
\end{document}